\documentclass[12pt]{article}
\usepackage[latin9]{inputenc}
\usepackage{amsmath}
\usepackage{amsthm}
\usepackage{amssymb}

\makeatletter

\newcommand*\LyXZeroWidthSpace{\hspace{0pt}}
\providecommand{\tabularnewline}{\\}

\theoremstyle{plain}
\newtheorem{thm}{\protect\theoremname}
\theoremstyle{definition}
\newtheorem{defn}[thm]{\protect\definitionname}
\theoremstyle{plain}
\newtheorem{fact}[thm]{\protect\factname}
\theoremstyle{theorem}
\newtheorem{claim}[thm]{\protect\claimname}
\theoremstyle{plain}
\newtheorem{lem}[thm]{\protect\lemmaname}
\theoremstyle{plain}
\newtheorem*{lem*}{\protect\lemmaname}
\theoremstyle{plain}
\newtheorem{prop}[thm]{\protect\propositionname}
\theoremstyle{theorem}
\newtheorem{remark}[thm]{\protect\remarkname}

\usepackage{fullpage}
\usepackage{amsthm}

\usepackage{bibspacing}

\usepackage{graphicx} % Required for inserting images
\usepackage{tikz}
\usetikzlibrary{trees,calc}
\usepackage{caption}
\usetikzlibrary{fit}
\tikzset{%
  highlight/.style={fill=red!30,fill,thick,inner sep=0pt}
}

\makeatother

\providecommand{\claimname}{Claim}
\providecommand{\definitionname}{Definition}
\providecommand{\factname}{Fact}
\providecommand{\lemmaname}{Lemma}
\providecommand{\propositionname}{Proposition}
\providecommand{\remarkname}{Remark}
\providecommand{\theoremname}{Theorem}

\begin{document}
\global\long\def\poly{\mathrm{{poly}}}%
\global\long\def\polylog{\mathrm{{polylog}}}%

\global\long\def\zo{\mathrm{\{0,1\}}}%

\global\long\def\mo{\mathrm{{-1,1}}}%

\global\long\def\e{\mathrm{\epsilon}}%

\global\long\def\d{\mathrm{\delta}}%

\global\long\def\a{\alpha}%

\global\long\def\b{\beta}%

\global\long\def\eps{\mathrm{\epsilon}}%

\global\long\def\E{\mathop{{}\mathbb{E}}}%

\global\long\def\F{\mathrm{\mathbb{F}}}%

\global\long\def\Z{\mathrm{\mathbb{Z}}}%

\global\long\def\R{\mathrm{\mathbb{R}}}%

\global\long\def\P{\mathrm{\mathbb{P}}}%

\global\long\def\sign{\mathrm{{Sign}}}%

\global\long\def\maj{\mathrm{{Maj}}}%

\title{Resilient functions: Optimized, simplified, and generalized}
\author{Peter Ivanov$^{*}$ \and Emanuele Viola\thanks{Partially supported by NSF grant CCF-2114116.}}
\maketitle
\begin{abstract}
An $n$-bit boolean function is resilient to coalitions of size $q$
if any fixed set of $q$ bits is unlikely to influence the function
when the other $n-q$ bits are chosen uniformly. We give explicit
constructions of depth-3 circuits that are resilient to coalitions
of size $cn/\log^{2}n$ with bias $n^{-c}$. Previous explicit constructions
with the same resilience had constant bias. Our construction is simpler
and we generalize it to biased product distributions. 

Our proof builds on previous work; the main differences are the use
of a tail bound for expander walks in combination with a refined analysis
based on Janson's inequality.
\end{abstract}

\section{Introduction}

A resilient function, informally speaking, is a function for which
a malicious adversary that controls a small coalition of the input
bits can not change the output with high probability. Resilient functions
have a wide range of applications. They were initially introduced
for \emph{coin flipping protocols} \cite{Ben-OrLi85,AL93,RZ98}. They
have also been used to construct \emph{randomness extractors} \cite{KampZ07,GVW-AC0ext,ChattopadhyayZ16,Meka17,ChengL16,HIV22}
and more recently, to show \emph{correlation bounds} against low-degree
$\mathbb{F}_{2}$ polynomials \cite{DBLP:conf/stoc/ChattopadhyayHH20}. 

Our main contribution is an improved and simplified construction that
generalizes to product distributions. Before we present our results
we introduce some definitions. 
\begin{defn}
\label{def:influence} Fix a function $f:\zo^{n}\to\zo$, a distribution
$D$ over $\zo^{n}$, and a coalition $Q\subseteq[n]$. Define $I_{Q,D}(f)$
to be the probability that $f$ is not fixed after the bits indexed
by $\overline{Q}$ are sampled according to $D$. When $D$ is the
uniform distribution we write $I_{Q}(f)$. 

We say $f$ is\emph{ $\rho$-tradeoff resilient} \emph{under} $D$
if for any $Q\subseteq[n],I_{Q,D}(f)\leq|Q|\rho$. We simply say $f$
is\emph{ $\rho$-tradeoff resilient }when $D$ is the uniform distribution.
\end{defn}

In this paper, every occurrence of $``c"$ denotes a possibly different
positive real number. Replacing $``c"$ with $O(1)$ everywhere is
consistent with one common interpretation of the big-Oh notation. 

First we provide non-explicit tradeoff resilient circuits.
\begin{thm}
\label{thm:resilient_balanced-main} For infinitely many $n$ there
exist monotone depth 3, size $n^{c}$ circuits $C:\zo^{n}\to\zo$
which are $\left(\frac{c\log^{2}n}{n}\right)$-tradeoff resilient
with bias $n^{-1+o(1)}.$
\end{thm}

This result is considered folklore, but we are not aware of any proofs
in the literature. Related works \cite{Ben-OrLi85,AL93,RZ98,ChattopadhyayZ16,Meka17,Wellens20}
either do not prove a full tradeoff, are not as balanced, or have
worse resilience. 

The construction in Theorem \ref{thm:resilient_balanced-main} is
due to Ajtai and Linial \cite{AL93}. In their seminal work, they
proved their construction is resilient to coalitions of size $cn/\log^{2}n$
which is close to the theoretical best $cn/\log n$ implied by the
KKL Theorem \cite{KahnKL88}. However, Ajtai and Linial did not prove
any tradeoff resilience. Russell and Zuckerman \cite{RZ98} then showed
the Ajtai-Linial construction has tradeoff resilience for coalitions
of size $n/\log^{c}n$, but not for smaller coalitions. There is a
more recent exposition \cite{Wellens20} on the topic; however, it
proves a tradeoff that is weaker. In particular, for coalitions of
constant size the resulting influence will be $c\log n/\sqrt{n}$
instead of $c\log^{2}n/n$. 

Next we state our main result, which is an explicit tradeoff resilient
function that generalizes to product distributions, improves the bias
to essentially match the non-explicit construction, and has a simplified
proof. 
\begin{defn}
$B_{\sigma}$ denotes the distribution over $\zo^{n}$ where each
bit is independently set to 1 with probability $\sigma$. 
\end{defn}

\begin{thm}
\label{thm:resilient_explicit} Fix integers $w\leq v\leq n$ where
$v$ is prime and $\sigma\in(0,1/2]$ s.t. $n=vw$ and $\sigma^{-w}=\frac{Cv}{\log v}$
for a fixed constant $C$. Then there are explicit monotone depth
3, size $n^{c}$ circuits $C:\zo^{n}\to\zo$ which are $\left(\frac{c\sigma^{-2}}{\log(\sigma^{-1})}\cdot\frac{\log^{2}n}{n}\right)$-tradeoff
resilient under $B_{\sigma}$ with bias $\sigma^{-2}/n^{1-o(1)}.$ 
\end{thm}

Let us provide some background on existing explicit resilient functions.
It is folklore that Majority is $c/\sqrt{n}$-tradeoff resilient.
Building on this, Ben-Or and Linial showed that Recursive-Majority-3
is $n^{-0.63..}$-tradeoff resilient \cite{Ben-OrLi85}.

More recently, Chattopadhyay and Zuckerman gave an\emph{ }explicit
$n^{-0.99}$-tradeoff resilient function with bias $n^{-c}$ \cite{ChattopadhyayZ16}.
Moreover, they derandomized the original Ajtai-Linial construction
so their function is computable by a small circuit. This was a key
part of their two-source extractor breakthrough. 

Meka \cite{Meka17} then improved the derandomization and achieved
$c\log^{2}n/n$-tradeoff resilience, but with constant bias. \cite{IMV23}
gave a generic way to compose tradeoff resilient circuits and this
allowed for nearly optimal size circuits with tradeoff resilience
$\log^{c}n/n$. And using this composition result, one can xor $c\log n$
copies of Meka's construction to achieve circuits with tradeoff resilience
$c\log^{3}/n$ and bias $n^{-c}$.

We remark all of the results stated above are over the uniform distribution.
A natural question is to consider resilience over non-uniform distributions
like product distributions. Such distributions arise naturally; for
instance, in a voting scenario one might consider the votes being
cast independently but not uniformly. 

The recent work of \cite{FHHHZ19} investigates this question by studying
how resilient \emph{any function }can be under arbitrary product distributions.
They show that in this setting, the KKL theorem essentially still
holds. Specifically, they prove a function can not be resilient to
coalitions of size $>cn\log\log n/\log n$ under any product distribution.

We complement their result by extending the Ajtai-Linial construction
to product distributions of the form $B_{\sigma}$, though with some
loss in resilience depending on $\sigma$. Its not clear whether this
dependence is necessary. Besides complementing \cite{FHHHZ19}, our
work might be useful for further improved resilient constructions;
see Section \ref{subsec:future-directions} for additional discussion.

Next we provide a table with known resilient circuit constructions
under $B_{1/2}$.

\LyXZeroWidthSpace{}

\begin{tabular}{|c|c|c|c|c|c|c|c|}
\hline 
Explicit: & Resilience: & Tradeoff: & Bias: & Monotone: & Depth: & Size: & Citation\tabularnewline
\hline 
Yes & $cn^{-1/2}$ & Yes & $0$ & Yes & $c\ensuremath{\log n}$ & $n^{c}$ & Folklore\tabularnewline
\hline 
Yes & $n^{-0.63...}$ & Yes & $0$ & Yes & $c\ensuremath{\log n}$ & $cn$ & \cite{Ben-OrLi85}\tabularnewline
\hline 
No & $c\log^{2}n/n$ & No & $n^{-c}$ & No & 3 & $cn^{2}$ & \cite{AL93}\tabularnewline
\hline 
No & $c\log^{2}n/n$ & Partial & $n^{-c}$ & No & 3 & $cn^{2}$ & \cite{RZ98}\tabularnewline
\hline 
No & $c\log^{2}n/n$ & Partial & $n^{-c}$ & No & 3 & $cn^{2}$ & \cite{Wellens20}\tabularnewline
\hline 
No & $c\log^{2}n/n$ & Yes & $n^{-c}$ & Yes & 3 & $n^{c}$ & Theorem \ref{thm:resilient_balanced-main}\tabularnewline
\hline 
Yes & $n^{-0.99}$ & Yes & $n^{-c}$ & Yes & 4 & $n^{c}$ & \cite{ChattopadhyayZ16}\tabularnewline
\hline 
Yes & $c\log^{2}n/n$ & Yes & $c$ & Yes & 3 & $n^{c}$ & \cite{Meka17}\tabularnewline
\hline 
Yes & $\log^{c}n/n$ & Yes & $n^{-c}$ & No & $c$ & $n^{1.01}$ & \cite{IMV23}\tabularnewline
\hline 
Yes & $c\log^{3}n/n$ & Yes & $n^{-c}$ & No & $5$ & $n^{c}$ & \cite{IMV23}\tabularnewline
\hline 
Yes & $c\log^{2}n/n$ & Yes & $n^{-c}$ & Yes & 3 & $n^{c}$ & Theorem \ref{thm:resilient_explicit}\tabularnewline
\hline 
\end{tabular}

\LyXZeroWidthSpace{}

Our next contribution is an explicit circuit that matches the KKL
Theorem up to constant factors and is exactly balanced under $B_{1/2}$.
As far as we know, only nearly balanced circuits were known prior
to this work.
\begin{thm}
\label{thm:match-kkl}For infinitely many $n$, there is an explicit
circuit $C:\zo^{n}\to\zo$ such that $I_{Q}(C)\leq c\log n/n$ for
any $Q\subseteq[n]:|Q|=1$ and $\E[C]=1/2$.
\end{thm}

To prove this we use a result by Ajtai and Linial which gives a generic
way to turn a nearly balanced circuit into an exactly balanced one,
without hurting the resilience too much. However, it comes at a cost
to the depth and monotonicity of the original circuit. 

\subsection{Proof of Theorem \ref{thm:resilient_explicit}}

The final construction is an Ajtai-Linial style AND of TRIBES circuit,
which will be defined through a generator $G$. The bulk of the proof
is showing that the resulting circuit will be 1) resilient and 2)
nearly balanced whenever $G$ is a sampler and a design. 1) follows
from the sampler property, while both properties are needed to prove
2). Our starting point is an expander walk which is a sampler by known
tail bounds. However, an expander walk does not form a good design
since two different walks can differ in just one node. Thus we pad
the expander walk with a short Reed-Solomon code which has good design
properties. This construction achieves a polynomial size domain which
results in a polynomial size circuit.

\paragraph*{Constructing a circuit through $G$}

First we set up notation. We identify $n$ with a $v\times w$ matrix
and for any $y\in[v]^{w}$ we associate a subset of $[n]$ of size
$w$ with one element per column in the natural way: the $k$th element
where $k\in[w]$ is in column $k$ and row $y[k]$. We let $S(y)$
denote this subset. 

From $S(y)$ we obtain other disjoint sets by increasing the row indices
by $j\bmod v$. We let $S(y,j)$ denote this set. Note that for any
$y\in[v]^{w}$ the sets $S(y,0),\dots,S(y,v-1)$ form a partition
of $[n]$ into $v$ sets of size $w$ each. 

We now define $C_{G}$. 
\begin{defn}
For any $G:[u]\to[v]^{w}$ and $i\in[u]$ we define 
\[
C_{G(i)}(x):=\vee_{j\in[v]}\wedge_{k\in S(G(i),j)}x_{k}
\]
 and we define 
\[
C_{G}(x):=\wedge_{i\in[u]}C_{G(i)}(x).
\]
\end{defn}

Note $C_{G}$ is an AND$_{u}$-OR$_{v}$-AND$_{w}$ circuit and $C_{G(i)}$
is a read-once OR$_{v}$-AND$_{w}$ circuit, where AND$_{w}$ denotes
a layer of AND gates of fan-in $w$, etc. For intuition see the following
illustration. 

\begin{figure}[!htb]
    \caption*{$G(i)$ defines a read-once OR$_v$-AND$_w$ subcircuit:}
    %\label{fig:enter-label}
    \hspace{-2em}
\begin{minipage}{0.4\textwidth}
    \caption*{\textcolor{red!60}{$S(G(i),0)$},\textcolor{green!60}{$S(G(i),1)$}, \textcolor{blue!60}{$S(G(i),v-1)$} }
    \caption*{\quad $w$}
    \label{fig:enter-label}
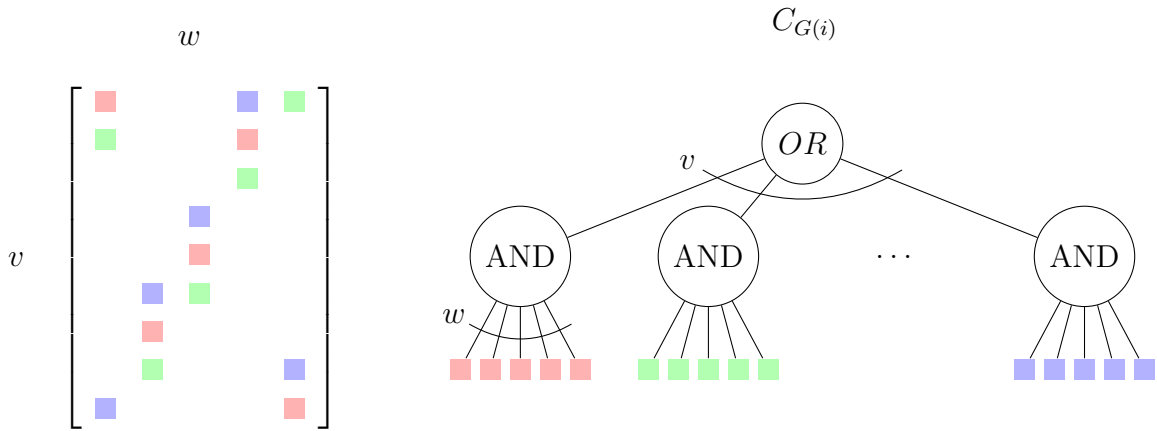
     \[
    v \quad \left[\begin{array}{*5{c}}
   \tikz[]{\node[fill=red!30]{}}&  &  &  \tikz[]{\node[fill=blue!30]{}} &\tikz[]{\node[fill=green!30]{}}  \\
        \tikz[]{\node[fill=green!30]{}}&   &   &  \tikz[]{\node[fill=red!30]{}} &  \\
      &   &  &  \tikz[]{\node[fill=green!30]{}}   &  \\
     &   & \tikz[]{\node[fill=blue!30]{}} &  &  \\
    &   &    \tikz[]{\node[fill=red!30]{}}&  &  \\
     &   \tikz[]{\node[fill=blue!30]{}} & \tikz[]{\node[fill=green!30]{}}  &  &  \\
    &    \tikz[]{\node[fill=red!30]{}} &  &  &  \\
     & \tikz[]{\node[fill=green!30]{}}   &  &  & \tikz[]{\node[fill=blue!30]{}} \\
     \tikz[]{\node[fill=blue!30]{}}  &   &  &    &   \tikz[]{\node[fill=red!30]{}} \\
  \end{array}\right]
    \]
\end{minipage}
\begin{minipage}{0.45\textwidth}
    \caption*{    \hspace{6em}$C_{G(i)}$}
    %\label{fig:enter-label}
\[   
    \begin{tikzpicture}[level 1/.style={sibling distance=25mm},level 2/.style={sibling distance=4mm}]
\node [circle,draw] (z){$OR$}
  child {node [circle,draw] (a) {AND}
    child {node [rectangle, fill=red!30] (b) {}
    }
      child {node [rectangle, fill=red!30] (cc) {}
    }
      child {node [rectangle, fill=red!30] (cc) {}
    }
      child {node [rectangle, fill=red!30] (dd) {}
    }
      child {node [rectangle, fill=red!30] (ddd) {}
    }
  }
    child {node [circle,draw] (bb) {AND}
    child {node [rectangle, fill=green!30] (ee) {}
    }
      child {node [rectangle, fill=green!30] (gg) {}
    }
      child {node [rectangle, fill=green!30] (hh) {}
    }
    child {node [rectangle, fill=green!30] (ii) {}
    }
      child {node [rectangle, fill=green!30] (jj) {}
    }
} 
    child {node [circle] (ff) {$\dots$} edge from parent[draw=none]
  }    
  child {node [circle,draw] (kk) {AND}
    child {node [rectangle, fill=blue!30] (ll) {}
    }
  child {node [rectangle, fill=blue!30] (mm) {}
  }
child {node [rectangle, fill=blue!30] (nn) {}
  }
child {node [rectangle, fill=blue!30] (oo) {}
  }
    child {node [rectangle, fill=blue!30] (pp) {}
  }
};
%\path (k) -- (l) node [midway] {+};
\path (cc) -- (cc) node [rectangle, fill=red!30] {};

\path (hh) -- (hh) node [rectangle, fill=green!30] {};

\path (nn) -- (nn) node [rectangle, fill=blue!30] {};

  child [grow=down] {
    %node (y) {$O()$}
    edge from parent[draw=none]
  };

\draw[black] ($(a)!0.6!(b) - (0.2, 0)$) to[bend left=-30]  ($(a)!0.6!(ddd)+ (0.2, 0)$);
\node[] at ($(a)!0.6!(b) - (0.4, -0.1)$  )  (a2)    {$w$};

\draw[black] ($(z)!0.3!(a) - (0.2, -0.1)$) to[bend left=-30]  ($(z)!0.3!(kk) + (0.2, 0.1)$);
\node[] at ($(z)!0.3!(a) - (0.4, -0.2)$  )  (a2)    {$v$};
  %node (w) [midway] {$O()$};
\end{tikzpicture}
\] 
\end{minipage}
\end{figure}

We next define some relevant quantities.
\begin{defn}
For integers $1\leq w\leq v\leq u$ and $0<\sigma<1$ we define 
\[
p:=(1-\sigma^{w})^{v},\quad bias(u,v,w,\sigma):=(1-p)^{u}.
\]
\end{defn}

Since the $C_{G(i)}(x)$ are read-once, we have $\P[C_{G(i)}(B_{\sigma})=0]=p$.
And if we supposed that $C_{G}$ was read-once (on $uvw$ bits), then
we would have $\P[C_{G}(B_{\sigma})=1]=bias(u,v,w,\sigma)$. Jumping
ahead, we will show that when $G$ is a sampler, $C_{G}$ on $vw$
bits behaves similarly. So by setting the parameters appropriately
and the following fact, $C_{G}$ will be nearly balanced. 
\begin{fact}
\label{fact:bias_independent} Fix $1\leq w\leq v\leq u$ and $0<\sigma<1$
so that $\sigma^{-w}\ln(u/\ln2)\leq v\leq\sigma^{-w}\ln(u/\ln2)+1$.
Then $|bias(u,v,w,\sigma)-1/2|\leq c\sigma^{w}$.
\end{fact}

For the sake of flow, we defer the proof of Fact \ref{fact:bias_independent}
and some of the other technical claims below to the appendix.

\paragraph*{The sampler property}
\begin{defn}
\label{def:sampler} $G:[u]\to[v]^{w}$ is a $(\a,\b)$\emph{-sampler}
if for any $f_{1},\dots,f_{w}:[v]\to\{0,1\}$ s.t. $\mu:=\E_{y\in[v]^{w}}[F(y)]\leq\frac{1}{\a\b\ln u}$
where $F:=\sum_{k\in[w]}f_{k}$, we have
\[
\E_{i\in[u]}[\a^{F(G(i))}]\leq e^{c\b\a\mu}.
\]
When $\b=c$ we say $G$ is an $\a$-sampler, and when $\a=2,\b=c$
we say $G$ is a sampler.
\end{defn}

The $1/\ln u$ factor is present in the definition for technical reasons.
We will also need a second version of the definition given next. 
\begin{claim}
\label{claim:equivalent-definition} The statement in Definition \ref{def:sampler}
is equivalent to the following:
\[
\E_{i\in[u]}[\a^{F(G(i))}\mathbf{1}_{F(G(i))\neq0}]\leq c\b\alpha\mu.
\]
\end{claim}

Next we show the identity function is a sampler.
\begin{fact}
\label{thm:sampler-intro-1-1}The identity function $I:[v]^{w}\to[v]^{w}$
is a sampler.
\end{fact}

This directly follows by the proceeding.
\begin{fact}
\label{fact:nonzero-to-zero} Let $X_{1},\dots,X_{w}$ be independent
$\{0,1\}$-valued r.v. s.t. $\mu:=\E[X]<1$ where $X:=\sum_{k\in[w]}X_{k}$.
Fix $\a>1$ s.t. $\a\mu\leq1$. Then $\E[\a^{X}]\leq e^{\a\mu},\E[\a^{X}\mathbf{1}_{X\neq0}]\leq2\a\mu$. 
\end{fact}

However, the identity function would not result in an efficient construction
since the resulting circuit would have size $\geq u=v^{w}=n^{c\log n}$.
Thus we are interested in samplers with a polynomial size domain $u=v^{c}$.
The existence of such samplers follows by the probabilistic method
(see Lemma \ref{lem:exists-shifted-sampler-1}). For an explicit construction,
we can take $G$ to be a random walk over a $(v,d,\lambda)$ expander,
which is a regular graph with $v$ nodes, degree $d$, and spectral
expansion $\lambda$. Tail bounds on expander walks \cite{RR17} imply
that $G$ is indeed a good sampler. 
\begin{thm}[\cite{RR17}]
\label{thm:sampler-walk} Let $G:[vd^{w}]\to[v]^{w}$ output walks
of length $w$ on a $(v,d,\lambda)$ expander for $\lambda<1/3$.
Then $G$ is an $\a$-sampler for any $1<\a<1/2\lambda$. 
\end{thm}

We remark that tail bounds from earlier works \cite{Lez98,Wagner08},
appropriately extended, would also suffice for our main result.

\paragraph{The design property}
\begin{defn}
We say $G$ is a $d$-design if $|S(G(i),j)\cap S(G(i'),j')|\leq w-d$
for any $i,i'\in[u],j,j'\in[v]$ s.t. $(i,j)\neq(i',j')$ .
\end{defn}

In other words, any two sets differ in at least $d$ elements. Note
$|S(G(i),j)\cap S(G(i),j')|=0$ for any $j\neq j'$ by definition.

We will require the design properties of the Reed-Solomon code:
\begin{fact}
\label{fact:rs-design} Fix some prime $v$ and integers $\ell\leq w\leq v$.
The degree $\ell$ Reed-Solomon code $RS:[v]^{\ell}\to[v]^{w}$ is
$\ell$-wise independent and a $w-\ell$ design.
\end{fact}

\paragraph*{The final construction}

The final generator is the concatenation of two different codes, as
is done in \cite{Meka17}. We replace the complicated extractor in
\cite{Meka17} with a standard expander walk, which grants the sampler
property by Theorem \ref{thm:sampler-walk}. 

The second code will be a constant-degree Reed-Solomon code of length
approximately $c\log\log n$. The final generator will then possess
the desired design property, albeit with a small loss in the sampler
property.
\begin{lem}
\label{thm:explicit-sampler-design} Fix integers $w\leq v$ where
$v$ is prime and $\sigma\in(0,1/2]$ s.t. $\sigma^{-w}=\frac{Cv}{\log v}$
for a fixed constant $C$. Then there is an explicit $G:[u]\to[v]^{w}$
that is a $(\sigma^{-1},\sigma^{-1})$-sampler and $4\log\log u/\log(\sigma^{-1})$-design,
where $u=\poly(v)$.
\end{lem}

\paragraph*{Proving resilience and small bias}

The remainder of the proof consists in showing that circuits obtained
from $G$ are 1) resilient and 2) nearly balanced. 1) was known for
$\sigma=1/2$, and it is straightforward to generalize to $\sigma\neq1/2$. 
\begin{lem}
\label{lem:resilience-1-1} Fix integers $w\leq v\leq u$ and $\sigma\in(0,1/2]$
so that $\sigma^{-w}=v/\ln(u/\ln2)$ and suppose $G:[u]\to[v]^{w}$
is a $(\sigma^{-1},\b)$-sampler. Then for any $Q\subseteq[n]$ s.t.
$|Q|\leq c\sigma^{-w+1}/\b$,
\[
I_{Q,B_{\sigma}}(C_{G})\leq|Q|\cdot c\b\sigma^{w-1}.
\]
\end{lem}

For 2), we follow an approach similar to that of \cite{ChattopadhyayZ16,Meka17},
which proved small bias by combining Janson's inequality with the
requirement that $G$ is a good design. However, we consider a \emph{slightly
tighter} \emph{version} of Janson's inequality; see Proposition \ref{lem:janson}
and the discussion there. This crucially allows us to apply the sampler
property of $G$ which in turn allows us to simplify the design requirements
on $G$ and improve the final bias to essentially match the nonexplicit
construction. 
\begin{lem}
\label{lem:sampler-to-bias-1} Suppose $G:[u]\to[v]^{w}$ is a $(\sigma^{-1},\b)$-sampler
and $d$-design. Then
\[
\left|\E[C_{G}(B_{\sigma})]-bias(u,v,w,\sigma)\right|\leq e^{\sigma^{d}\log^{4}u}\frac{\b\sigma^{-1}}{n^{1-o(1)}}.
\]
\end{lem}

We are now ready to prove Theorem \ref{thm:resilient_explicit}.
\begin{proof}[Proof of Theorem \ref{thm:resilient_explicit}. ]

Take $G$ from Lemma \ref{thm:explicit-sampler-design}. By Lemma
\ref{lem:resilience-1-1},
\[
I_{Q,B_{\sigma}}(C_{G})\leq|Q|\cdot c\sigma^{-2}\sigma^{w}=|Q|\cdot\frac{c\sigma^{-2}}{\log(\sigma^{-1})}\cdot\frac{\log^{2}n}{n}.
\]
The last inequality follows since $\sigma^{-w}=cv/\log v$ thus $w=c\log v/\log(\sigma^{-1})$
and $v=n/w$ so $\log v=c\log n$. 

By Lemma \ref{lem:sampler-to-bias-1}, $\left|\E[C_{G}(B_{\sigma})]-bias(u,v,w,\sigma)\right|\leq\sigma^{-2}/n^{1-o(1)}$.
We conclude by Fact \ref{fact:bias_independent}.
\end{proof}

\subsection{Future Directions \label{subsec:future-directions}}

A long-standing open problem is to improve the $cn/\log^{2}n$ resilience
achieved by Ajtai and Linial. In this direction, we pose the following
question: Is there a function $f:\zo^{n}\to\zo$ such that 1) $\P[f=0]=1/n$
and 2) the influence of each bit is $c\log n/n^{2}$? In other words,
are there biased functions matching the KKL theorem? 

It is well-known there are balanced functions which match KKL, namely
TRIBES. However, when the parameters are set so that $\P[\text{TRIBES}=0]=1/n$,
the influence of each bit becomes $c\log^{2}n/n^{2}$. The Ajtai-Linial
construction is an AND over such biased TRIBES. 

Thus any progress on the question above can be viewed as a first step
towards beating $cn/\log^{2}n$ resilience. The works of \cite{EG20,EKLM22}
provide some information on the structure of functions which match
the KKL Theorem that may be helpful.

Another problem is to obtain alternative constructions that improve
on \cite{Ben-OrLi85}. One possible approach follows from Theorem
\ref{thm:resilient_explicit} by setting $w$ constant (and a suitable
$\sigma$) which allows one to reach influence $n^{-0.99}$; alternative
derandomizations of the Ajtai-Linial construction in this setting
could be of interest.

\subsection{Organization}

In Section \ref{sec:non-explicit} we prove Theorem \ref{thm:resilient_balanced-main}.
We prove Lemmas \ref{thm:explicit-sampler-design}, \ref{lem:resilience-1-1},
\ref{lem:sampler-to-bias-1} in Sections \ref{sec:explicit}, \ref{subsec:exist-resilient},
\ref{sec:sampler-to-bias} respectively. In Section \ref{sec:al-trick}
we prove Theorem \ref{thm:match-kkl}.

\section{Nonexplicit tradeoff resilient circuits \label{sec:non-explicit}}

Here we prove the existence of good samplers and designs over the
uniform distribution.
\begin{lem}
\label{lem:exists-shifted-sampler-1} For large enough integers $w\leq v$
there is a $G:[v^{2}2^{2w}\ln v]\to[v]^{w}$ that is a sampler and
$\lfloor w/2\rfloor$-design.
\end{lem}

Combining this with Lemmas \ref{lem:resilience-1-1} and \ref{lem:sampler-to-bias-1}
we can prove Theorem \ref{thm:resilient_balanced-main}.
\begin{proof}[Proof of Theorem \ref{thm:resilient_balanced-main}]

Fix $w$ and set $v=\lceil2^{w}\ln(u/\ln2)\rceil$ where $u:=(vw)^{4}$
and set $n=vw$. Note $u\geq v^{2}2^{2w}\ln v.$ Take $G$ from Lemma
\ref{lem:exists-shifted-sampler-1}. By Lemma \ref{lem:resilience-1-1},
\[
I_{Q}(C_{G})\leq c|Q|\cdot2^{-w+1}\leq c|Q|\cdot\frac{\log^{2}n}{n}.
\]
The last inequality follows since $v=c2^{w}\log n$ so $n=c2^{w}w\log n$
which implies that $2^{w}=cn/\log^{2}n$. And by Lemma \ref{lem:sampler-to-bias-1}
and Fact \ref{fact:bias_independent}, $|\E[C_{G}]-1/2|\leq n^{-1+o(1)}.$
\end{proof}
Let us add a couple of remarks. First, one can improve the final circuit
size to $cn^{2}$ through a different proof which avoids the sampler
definition and argues directly about each subcircuit. However, we
choose to present the current proof as it is simpler and more consistent
with the explicit construction. Second, one can replace Lemma \ref{lem:sampler-to-bias-1},
the proof of which is somewhat involved, with a simpler argument due
to \cite{Wellens20}, at the cost of the monotonicity of the final
nonexplicit circuit.

\subsection{Proof of Lemma \ref{lem:exists-shifted-sampler-1} }

For any $G:[u]\to[v]^{w}$ and $F=\sum_{k\in[w]}f_{k}$, where $f_{1},\dots,f_{w}:[v]\to\{0,1\}$,
let $S_{F}(G):=\sum_{i\in[u]}2^{F(G(i))}\mathbf{1}_{F(G(i))\neq0}$.
By Fact \ref{fact:nonzero-to-zero}, over a uniformly sampled $G:[u]\to[v]^{w}$
and any fixed $F$ with mean $\E_{y\in[v]^{w}}[F(y)]=\mu<1/2$ we
have
\[
\E_{G}\left[S_{F}(G)\right]\leq u\cdot4\mu.
\]
By the above and Hoeffding's inequality, for any fixed $F$ with mean
$\mu=t/v$ we have

\[
\P_{G}\left[S_{F}(G)-4u\mu>2u\mu\right]<exp\left(-\frac{2(2u\mu)^{2}}{u2^{2w}}\right)<exp\left(-\frac{8ut^{2}}{v^{2}2^{2w}}\right)<v^{-8t^{2}}.
\]
The last $<$ follows for $u\geq v^{2}2^{2w}\ln v$. Now we union
bound the probability of some bad $F$ with mean $t/v$:
\begin{align*}
\P_{G}[\exists F:\E[F]=t/v\wedge S_{F}(G)>6u\mu]<|\{F:\E[F]=t/v\}|\cdot v^{-8t^{2}} & \leq v^{-6t}.
\end{align*}
The last $\leq$ follows by the Vandermonde identity, which says there
are $\sum_{\ell_{1}+\dots+\ell_{w}=t}\binom{v}{\ell_{1}}\dots\binom{v}{\ell_{w}}=\binom{vw}{t}\leq(vw)^{t}$
functions $F$ with mean $t/v$. 

After a union bound over all possible $1\leq t\leq v$, the probability
of some bad $F$ with mean $\leq1$ is at most $\sum_{t=1}^{\infty}v^{-6t}<2v^{-6}<1/2$.

Now we bound the probability that $G$ is a $w/2$ design. For a fixed
$z\in[v]^{w}$ and a uniformly sampled $y\sim[v]^{w}$, the probability
$y$ intersects $z$ in more than $w/2$ elements is $\leq(c/v)^{w/2}$.
Thus by a union bound, 
\[
\P_{G}[\exists i'\neq i\in[u]\wedge j',j\in[v]:|S(G(i'),j')\cap S(G(i),j)|\geq w/2]\leq(uv)^{2}(c/v)^{w/2}<1/2.
\]
The last $<$ holds for $w$ a large enough constant since $u=(vw)^{4},w\leq v$.
Hence there is a $G:[u]\to[v]^{w}$ that is a sampler and $(w/2)$-design
for $u\geq v^{2}2^{2w}\ln v$. $\hfill\qed$

\section{Explicit tradeoff resilient circuits \label{sec:explicit}}

In this section we we prove Lemma \ref{thm:explicit-sampler-design},
restated below, which provides an explicit $\sigma^{-1}$-sampler
and design. The bulk of the output of $G:[u]\to[v]^{w}$ will come
from an expander walk, but a small subset of the output of size roughly
$\log w$ will come from a constant-degree Reed-Solomon code. 

To analyze the sampler property of $G$ (Proposition \ref{prop:G-is-sampler})
we apply Cauchy-Schwarz and use the known sampler properties of expander
walks (Theorem \ref{thm:sampler-walk}) and Reed-Solomon codes (Proposition
\ref{prop:cwise-sampler}). However, there is a $\sigma^{-1}$ factor
loss from applying Cauchy-Schwarz.

$G$ will be approximately a $c\log w$-design (Proposition \ref{prop:G-is-design})
which simply follows from the design properties of the Reed-Solomon
code.
\begin{lem*}
\label{thm:explicit-sampler-design-1} Fix $w\leq v$ where $v$ is
prime and $\sigma\in(0,1/2]$ s.t. $\sigma^{-w}=\frac{Cv}{\log v}$
for a fixed constant $C$. Then there is an explicit $G:[u]\to[v]^{w}$
that is a $(\sigma^{-1},\sigma^{-1})$-sampler and $4\log\log u/\log(\sigma^{-1})$-design,
where $u=\poly(v)$.
\end{lem*}

\subsection{Proof of Lemma \ref{thm:explicit-sampler-design}}

Fix $w\leq v$ where $v$ is prime and $\sigma\in(0,1/2]$ so that
$\sigma^{-w}=v/\ln(u/\ln2)$ for $u$ which we specify below.

It is well known there are explicit $(v,d,\lambda)$ expanders with
$d\leq\lambda^{-c}$ for any explicit $\lambda$. For instance, one
can take powers of constant degree expanders with constant expansion.
Let $W$ output walks of length $w_{1}\leq w$ on such an expander
with $\lambda=\sigma^{2}/4$ for some $w_{1}$ we specify later. By
Theorem \ref{thm:sampler-walk}, $W:[vd^{w_{1}}]\to[v]^{w}$ is a
$\sigma^{-2}$-sampler. Note $vd^{w_{1}}=v^{c}$ since $d=\sigma^{-c}$
and $\sigma^{-w}\leq v$. 

Let $RS:[v]^{c_{1}}\to[v]^{w_{2}}$ denote the code from Fact \ref{fact:rs-design}
where $c_{1}$ is a constant large enough so that $v^{c_{1}}\geq vd^{w_{1}}$
and $w_{2}:=(c_{1}/3)\max(\lfloor\log\ln u/\log(\sigma^{-1})\rfloor,4)$.
Note $w_{2}<w$ since $\sigma^{-w}=v/\ln(u/\ln2)$ and $c_{1}<w_{2}$. 

Finally, we set $w_{1}:=w-w_{2}$ and $u:=vd^{w_{1}}$ which implies
$v/\ln(u/\ln2)=(1+o(1))Cv/\log v$ for some fixed constant $C$. We
define $G:[u]\to[v]^{w}$ as follows:
\[
\begin{cases}
G(i)_{k}=W(i)_{k} & \text{if }1\leq k\leq w_{1}\\
G(i)_{k}=RS(i)_{k-w_{1}} & \text{if }w_{1}<k\leq w.
\end{cases}
\]

\begin{prop}
\label{prop:G-is-sampler}$G$ is a $(\sigma^{-1},\sigma^{-1})$-sampler.
\end{prop}

\begin{proof}
Fix $f_{1},\dots,f_{w}:[v]\to\{0,1\}$ so that $\mu=\E[F]\leq1/(c\sigma^{-2}\ln u)$
where $F=\sum_{k\in[w]}f_{k}$. Let $F_{1}=f_{1}+\dots+f_{w_{1}}$,
$F_{2}=f_{w_{1}+1}+\dots+f_{w}$, $\mu_{1}=\E[F_{1}],\mu_{2}=\E[F_{2}]$
(so $\mu_{1}+\mu_{2}=\mu$). Then
\begin{align*}
\E_{i\in[u]}[\sigma^{-F(G(i))}] & =\E_{i\in[u]}[\sigma^{-F_{1}(W(i))}\sigma^{-F_{2}(RS(i))}]\\
 & \leq\E_{i\in[u]}[\sigma^{-2F_{1}(W(i))}]^{1/2}\E_{i'\in[u]}[\sigma^{-2F_{2}(RS(i'))}]^{1/2}\\
 & \leq e^{c\sigma^{-2}\mu_{1}}\cdot(e^{c\sigma^{-2}\mu_{2}}+\sigma^{-2w_{2}}\mu_{2}^{c_{1}})^{1/2}.
\end{align*}
The second $\leq$ follows by Cauchy-Schwarz. The last $\leq$ follows
since $W$ is a $\sigma^{-2}$ sampler and by Proposition \ref{prop:cwise-sampler},
stated at the end. To conclude it suffices to show
\[
(\sigma^{-2w_{2}}\mu_{2}^{c_{1}})^{1/2}\leq\mu_{2}.
\]
This follows as $\sigma^{-2w_{2}}\leq(\ln u)^{(2c_{1}/3)}$ and $\mu_{2}^{c_{1}}\leq\mu^{(2c_{1}/3)}\mu_{2}^{(c_{1}/3)}\leq(\ln u)^{-(2c_{1}/3)}\mu_{2}^{(c_{1}/3)}$. 
\end{proof}
\begin{prop}
\label{prop:G-is-design} $G$ is a $4\log\log u/\log(\sigma^{-1})$-design.
\end{prop}

\begin{proof}
Its clear that $G$ is a $(w_{2}-c_{1})$-design, and note $w_{2}-c_{1}\geq(c_{1}/6)\cdot(\log\ln u/\log(\sigma^{-1}))\geq4\log\log u/\log(\sigma^{-1})$
for $c_{1}$ large enough.
\end{proof}

\subsection{Remaining proof}

The following is a sampler like property of Reed-Solomon codes.
\begin{prop}
\label{prop:cwise-sampler} Let $D$ be a $\ell$-wise uniform distribution
on $[v]^{w}$. Then for any $f_{1},\dots,f_{w}:[v]\to\zo$ s.t. $\mu:=\E_{y\in[v]^{w}}[F(y)]\leq1/2\a$
where $F=\sum_{k\in[w]}f_{k}$,
\[
\E[\a^{F(D)}]\leq e^{c\a\mu}+\a^{w}\mu^{\ell}.
\]
\end{prop}

\begin{proof}
Define the random variables $Y_{1}=f_{1}(y_{1}),\dots Y_{w}=f_{w}(y_{w})$
where $y$ is sampled from $D$ and let $Y=\sum_{k\in[w]}Y_{k}$.
Then 
\[
\P[Y\geq\ell]\leq\sum_{S\subseteq[w],|S|=\ell}\P\left[\prod_{i\in S}Y_{i}=1\right]=\sum_{S\subseteq[w],|S|=\ell}\prod_{i\in S}\P[Y_{i}=1]\leq\binom{w}{\ell}\left(\frac{\mu}{w}\right)^{\ell}\leq\mu^{\ell}.
\]
The $=$ follows as $D$ is $\ell$-wise uniform. The next $\leq$
follows by Maclaurin's inequality. Since the above holds for any $k\leq\ell$,
we have
\begin{align*}
\E[\a^{Y}] & \leq\E[\a^{Y}|Y=0]+\sum_{k=1}^{\ell-1}\E[\a^{Y}|Y=k]\P[Y=k]+\E[\a^{Y}|Y\geq\ell]\P[Y\geq\ell]\\
 & \leq1+\sum_{k=1}^{\ell-1}\a^{k}\mu{}^{k}+\a^{w}\mu^{\ell}\\
 & \leq1+2\a\mu+\a^{w}\mu^{\ell}.
\end{align*}
The last $\leq$ follows since $\alpha\mu<1/2$. We can now conclude
since $1+x\leq e^{x}$.
\end{proof}

\section{Sampler to resilience \label{subsec:exist-resilient}}

In this section we prove Lemma \ref{lem:resilience-1-1} which says
that when $G$ is a sampler, then the resulting circuit $C_{G}$ will
be resilient. To do so we bound the probability that each read-once
subcircuit $C_{G(i)}$ will be fixed and then apply a union bound.
Each subcircuit $C_{G(i)}$ is unlikely to be fixed since the coalition
size is less than the number of independent $AND_{w}$ terms in $C_{G(i)}$. 
\begin{lem*}
\label{lem:resilience-1} Fix integers $w\leq v\leq u$ and $\sigma\in(0,1/2]$
so that $\sigma^{-w}=v/\ln(u/\ln2)$ and suppose $G:[u]\to[v]^{w}$
is a $(\sigma^{-1},\b)$-sampler. Then for any $Q\subseteq[n]$ s.t.
$|Q|\leq c\sigma^{-w+1}/\b$,
\[
I_{Q,B_{\sigma}}(C_{G})\leq|Q|\cdot\b\sigma^{w-1}.
\]
\end{lem*}
\begin{proof}
Fix a coalition $Q\subseteq[n]$ of size $q$. After sampling from
$B_{\sigma}$ the bits indexed by $\overline{Q}$, $C_{G(i)}$ is
not fixed iff $E_{i}\wedge F_{i}=1$ where

\[
\begin{cases}
E_{i}=1 & \text{if }\forall j:S(G(i),j)\cap Q=\emptyset,A_{i,j}(x)=0,\\
F_{i}=1 & \text{if }\exists j:S(G(i),j)\cap Q\neq\emptyset,A_{i,j}(x)=\,?
\end{cases}
\]
where $A_{i,j}(x):=\wedge_{k\in S(G(i),j)}x_{k}$. By $A_{i,j}(x)=\,?$
we denote that the bits indexed by $S(G(i),j)\cap\overline{Q}$ are
set to 1. 

First we bound $\P[E_{i}]$. Since there are $v$ AND$_{w}$ terms
in $C_{G(i)}$ and $Q$ can intersect with $\leq q$ of them we have
\begin{align*}
\P[E_{i}] & \leq(1-\sigma^{w})^{v-q}\le p\cdot e^{q\sigma^{w}}\leq c/u.
\end{align*}
since $v=c\sigma^{-w}\ln u$ and $q\leq c\sigma^{-w}$.

Next we bound $\P[F_{i}]$ by union bounding over all intersecting
$j$:
\begin{align*}
\P[F_{i}] & \leq\sum_{j\in[v]:S(G(i),j)\cap Q\neq0}\sigma^{(w-|S(G(i),j)\cap Q|)}=\sigma^{w}\sum_{j\in[v]}\sigma^{-|S(G(i),j)\cap Q|}\mathbf{1}_{|S(G(i),j)\cap Q|\neq0}.
\end{align*}
Combining the bounds above we have
\begin{align*}
\sum_{i\in[u]}I_{Q}(C_{P_{i}}) & \leq\sum_{i\in[u]}\P[E_{i}]\P[F_{i}]\\
 & \leq\frac{c}{u}\sigma^{w}\sum_{i\in[u]}\sum_{j\in[v]}\sigma^{-|S(G(i),j)\cap Q|}\mathbf{1}_{|S(G(i),j)\cap Q|\neq0}\\
 & =c\sigma^{w}\sum_{j\in[v]}\E_{i\in[u]}\sigma^{-|S(G(i),j)\cap Q|}\mathbf{1}_{|S(G(i),j)\cap Q|\neq0}.
\end{align*}
To conclude we claim that for any fixed $j\in[v]$,
\[
\E_{i\in[u]}\sigma^{-|S(G(i),j)\cap Q|}\mathbf{1}_{|S(G(i),j)\cap Q|\neq0}\leq\b\sigma^{-1}q/v.
\]
To see this, for each $k\in[w]$ define $f_{k}:[v]\to\{0,1\}$ as
\[
\begin{cases}
f_{k}(y)=1 & \text{if }\{((k-1)v+y+j)\bmod v\}\in Q;\\
f_{k}(y)=0 & \text{if }\{((k-1)v+y+j)\bmod v\}\notin Q.
\end{cases}
\]
Note $|S(G(i),j)\cap Q|=F(G(i))$ and $\E_{y\in[v]^{w}}[F(y))]=q/v\leq c\sigma^{-w+1}/(v\b)\leq1/(\sigma^{-1}\b\ln u)$
where $F=\sum_{k\in[w]}f_{k}$. The claim now follows since $G$ is
a $(\sigma^{-1},\b)$ sampler.
\end{proof}

\section{Sampler and design to small bias \label{sec:sampler-to-bias}}

Here we prove Lemma \ref{lem:sampler-to-bias-1}, restated below.
\begin{lem*}
Suppose $G:[u]\to[v]^{w}$ is a $(\sigma^{-1},\b)$-sampler and a
$d$-design. Then 
\[
\left|\E[C_{G}(B_{\sigma})]-bias(u,v,w,\sigma)\right|\leq e^{\sigma^{d}\cdot\log^{4}u}\frac{\b\sigma^{-1}}{n^{1-o(1)}}.
\]
\end{lem*}
For this we will need two different approximations of the OR of boolean
circuits. The first is the Bonferroni inequality.
\begin{prop}
\label{fact:bonferroni}For any boolean random variables $Z_{1},\dots,Z_{n}$
and odd $K$, letting $Z:=\vee_{i\in[n]}Z_{i}$ we have
\[
0\leq\P[Z]-\sum_{k\in[K-1]}(-1)^{k-1}S_{k}(Z_{1},\dots,Z_{n})\leq S_{K}(Z_{1},\dots,Z_{n})
\]
where $S_{k}(Z_{1},\dots,Z_{n}):=\sum_{S\subseteq[n]:|S|=k}\P[\wedge_{i\in S}Z_{i}]$.
\end{prop}

We also need the following slight tightening of Janson's inequality
\cite{Jan90}. The below version is implicit in the proof presented
by \cite{AlonSpEr92} (c.f. \cite{LSSUV21}). We provide a justification
in the appendix.
\begin{prop}
\label{lem:janson} Let $C_{1},\dots,C_{n}$ be arbitrary monotone
boolean circuits such that $\P[C_{i}=0]=p\ge1/2$ $\forall i$ over
some product distribution $D$. Then letting $C:=\vee_{i\in[n]}C_{i}$
we have
\[
\prod_{i\in[n]}\P[C_{i}=0]\leq\P[C=0]\leq\prod_{i\in[n]}\P[C_{i}=0]\cdot\left(1+\sum_{\ell=1}^{n}\frac{2^{\ell}}{\ell!}\Delta(C)^{\ell}\right)
\]
where the probabilities are over $D$ and 
\begin{align*}
\Delta(C):=\sum_{i\in[n]}L(C,i),\quad L(C,i):=\sum_{\substack{i'\in[n]:\\
i'<i,\\
i'\sim i
}
}\P[C_{i'}\wedge C_{i}]
\end{align*}
where $i'\sim i$ denotes that $C_{i'}$ and $C_{i}$ are not on disjoint
variables.
\end{prop}

The upper bound is usually stated as $\prod_{i\in[n]}\P[C_{i}=0]\cdot e^{2\Delta(C)}$.
If one used this bound, then in the course of proving Lemma \ref{lem:sampler-to-bias-1}
one needs to bound the quantity $\E_{\substack{T\subseteq[u]:|T|=k}
}e^{\Delta(C_{G(T)})}-1$ for $k$ not too large, where $C_{G(T)}:=\vee_{i\in T}C_{G(i)}$.
\cite{ChattopadhyayZ16,Meka17} do so by requiring design properties
of $G$. If $G$ is a $d$-design, then $\E_{\substack{T\subseteq[u]:|T|=k}
}e^{\Delta(C_{G(T)})}\leq e^{\log^{c}u/2^{d}}$ (see Lemma \ref{cor:max-power}), resulting in a final bias of approximately
$2^{-d}$. To achieve polynomial bias, one would need a Reed-Solomon
code of length $c\log n$, which would ruin the sampler property.

Instead if one uses the above version of Janson's inequality, one
needs to bound $\E_{\substack{T\subseteq[u]:|T|=k}
}\Delta(C_{G(T)})^{\ell}$ for $1\leq\ell\leq n$. Through the sampler property of $G$, we
prove $\E_{\substack{T\subseteq[u]:|T|=k}
}\Delta(C_{G(T)})=n^{-1+o(1)}$ (Lemma \ref{lem:post-janson-1}). We then use Lemma \ref{cor:max-power}
to bound the remaining terms $\E\sum\Delta(C_{G(T)})^{\ell}$. This
analysis results in a final bias of $e^{\log^{c}u/2^{d}}n^{-1+o(1)}=n^{-1+o(1)}$
for $d$ roughly $c\log\log n$.
\begin{lem}
\label{lem:post-janson-1}Suppose $G:[u]\to[v]^{w}$ is a $(\sigma^{-1},\b)$-sampler.
For any $1\leq k\leq\log u$,

\[
\E_{\substack{T\subseteq[u]:\\
|T|=k
}
}\Delta(C_{G(T)})\leq\frac{\b\sigma^{-1}}{v^{1-o(1)}}.
\]
\end{lem}

\begin{lem}
\label{cor:max-power}Suppose $G:[u]\to[v]^{w}$ is a $d$-design.
For any $1\leq k\leq\log u$,

\[
\max_{\substack{T\subseteq[u]:\\
|T|=k
}
}\Delta(C_{G(T)})\leq\log^{4}u\cdot\sigma^{d}.
\]
\end{lem}

Assuming Lemmas \ref{lem:post-janson-1} and \ref{cor:max-power}
we can prove Lemma \ref{lem:sampler-to-bias-1}.
\begin{proof}[Proof of Lemma \ref{lem:sampler-to-bias-1} ]

We can write $\neg C_{G}=\vee_{i\in[u]}Z_{i}$ where $Z_{i}:=\neg C_{G(i)}$.
We can also write $1-bias(u,v,w,\sigma)=1-(1-p)^{u}$ where recall
$p=(1-\sigma^{w})^{v}$. By Bonferroni's inequality, for any odd $K$
which we set later we have
\begin{align*}
\left|\E[\neg C_{G}]-\sum_{k\in[K-1]}(-1)^{k-1}\E[S_{k}(Z_{1},\dots,Z_{u})]\right| & \leq\left|\E[S_{K}(Z_{1},\dots,Z_{u})\right|,\\
\left|(1-bias(u,v,w,\sigma))-\sum_{k\in[K-1]}(-1)^{k-1}\binom{u}{k}p^{k}\right| & \leq\binom{u}{K}p^{K}
\end{align*}
where the expectations are over $B_{\sigma}$. Note that $\E[S_{k}(Z_{1},\dots,Z_{u})]=\sum_{T\subseteq[u]:|T|=k}\P[C_{G(T)}=0]$,
where $C_{G(T)}=\vee_{i\in T}C_{G(i)}=\vee_{i\in T}\vee_{j\in[v]}A_{i,j}$,
and $A_{i,j}=\wedge_{k\in S(G(i),j)}x_{k}$. Thus we view $C_{G(T)}$
as an OR$_{kv}$-AND$_{w}$ circuit. First note that by Lemmas \ref{lem:post-janson-1},
\ref{cor:max-power} we have
\[
\E_{\substack{T\subseteq[u]:\\
|T|=k
}
}\Delta(C_{G(T)})^{\ell}\leq\max_{\substack{T\subseteq[u]:\\
|T|=k
}
}\Delta(C_{G(T)})^{\ell-1}\E_{\substack{T\subseteq[u]:\\
|T|=k
}
}\Delta(C_{G(T)})\leq(\log^{4}u\cdot\sigma^{d})^{\ell-1}\binom{u}{k}\frac{\b\sigma^{-1}}{v^{1-o(1)}}.
\]
Combining this with Proposition \ref{lem:janson} we have
\begin{align*}
\binom{u}{k}p^{k}\leq\E[S_{k}(Z_{1},\dots,Z_{u})] & \leq p^{k}\sum_{T\subseteq[u]:|T|=k}(1+\sum_{\ell=1}^{kv}\frac{2^{\ell}}{\ell!}\Delta(C_{G(T)})^{\ell})\\
 & \leq\binom{u}{k}p^{k}(1+\frac{\b\sigma^{-1}}{v^{1-o(1)}}\sum_{\ell=1}^{kv}\frac{2^{\ell}}{\ell!}(\log^{4}u\cdot\sigma^{d})^{\ell-1})\\
 & \leq\binom{u}{k}p^{k}(1+\frac{\b\sigma^{-1}}{v^{1-o(1)}}e^{\log^{4}u\cdot\sigma^{d}}).
\end{align*}

The last inequality follows since $\sum_{\ell=0}^{\infty}x^{\ell}/\ell!=e^{x}$.
If the $Z_{1},\dots,Z_{u}$ were independent then of course $\E[S_{k}(Z_{1},\dots,Z_{u})]=\binom{u}{k}p^{k}.$
So the above is saying that when $G$ is a sampler, the resulting
$Z_{i}$ behave as if they were independent up to a small multiplicative
error.

After repeated applications of the triangle inequality, letting $\d=\frac{\b\sigma^{-1}}{v^{1-o(1)}}e^{\log^{4}u\cdot\sigma^{d}}$,
we have 
\begin{align*}
\left|\E[C_{G}]-bias(u,v,w,\sigma)\right| & \leq\d\sum_{k=1}^{K-1}\binom{u}{k}p^{k}+(2+\d)\binom{u}{K}p^{K}\\
 & \leq\d(1+p)^{u}+(c/K)^{K}\\
 & \leq c\d.
\end{align*}
The last inequality follows since $p=c/u$ and by setting $K=c\log v/\log\log v$.
\end{proof}

\subsection{Proof of Lemmas \ref{lem:post-janson-1}, \ref{cor:max-power}}

Fix any $T\subseteq[u]:|T|=k$ and define $Y(i',j',i,j):=\sigma^{-|S(G(i'),j')\cap S(G(i),j)|}\mathbf{1}_{|S(G(i'),j')\cap S(G(i),j)|\neq0}$.
By definition we have
\begin{align*}
\Delta(C_{G(T)}) & =\sum_{(i,j)\in T\times[v]}L(C_{G(T)},(i,j))\\
 & =\sum_{(i,j)\in T\times[v]}\sum_{\substack{(i',j')\in T\times[v]:\\
(i',j')<(i,j),\\
(i',j')\sim(i,j)
}
}\P[A_{i',j'}\wedge A_{i,j}]\\
 & =\sigma^{2w}\sum_{(i,j)\in T\times[v]}\sum_{\substack{(i',j')\in T\times[v]:\\
(i',j')<(i,j)
}
}Y(i',j',i,j)
\end{align*}

The last equality follows as $\P[A_{i',j'}\wedge A_{i,j}]=\sigma^{2w-|S(G(i'),j')\cap S(G(i),j)|}$
and $\mathbf{1}_{(i',j')\sim(i,j)}=\mathbf{1}_{|S(G(i'),j')\cap S(G(i),j)|\neq0}$.
One can think of $(i,j)$ as a number, with $i$ as the most significant
bit. 

To prove Lemma \ref{cor:max-power}, since $G$ is a $d$-design,
for any fixed $i\neq i'\in T$ and $j\in[v]$ we have
\[
\sum_{\substack{j'\in[v]:\\
(i',j')<(i,j)
}
}Y(i',j',i,j)\leq w\sigma^{-(w-d)}.
\]
Thus $\Delta(C_{G(T)})\leq\sigma^{2w}\cdot k^{2}v\cdot w\sigma^{-(w-d)}\leq\sigma^{d}\cdot k^{2}\ln u\cdot w\leq\sigma^{d}\ln^{4}u$.
The second inequality follows since $v=c\sigma^{-w}\ln u$.

Now we prove Lemma \ref{lem:post-janson-1}. We need the following
result.
\begin{prop}
\label{prop:powers-to-sampler} Suppose $G:[u]\to[v]^{w}$ is a $(\sigma^{-1},\b)$-sampler.
For any $1\leq k\leq u$ and $j',j\in[v]$,
\begin{align*}
\E_{\substack{T\subseteq[u]:\\
|T|=k
}
}\sum_{\substack{i',i\in T:\\
i'<i
}
}Y(i',j',i,j) & \leq\binom{u}{k}k^{2}\frac{\b\sigma^{-1}}{v^{1-o(1)}}\cdot
\end{align*}
\end{prop}

Assuming Proposition \ref{prop:powers-to-sampler}, we can prove Lemma
\ref{lem:post-janson-1}.
\begin{align*}
\E_{\substack{T\subseteq[u]:\\
|T|=k
}
}\Delta(C_{G(T)}) & \leq\sigma^{2w}\E_{\substack{T\subseteq[u]:\\
|T|=k
}
}\sum_{\substack{(i,j)\in T\times[v]}
}\sum_{\substack{(i',j')\in T\times[v]:\\
(i',j')<(i,j)
}
}Y(i',j',i,j)\\
 & =(\frac{c\ln u}{v})^{2}\sum_{j',j\in[v]}\E_{\substack{T\subseteq[u]:\\
|T|=k
}
}\sum_{\substack{i',i\in T:\\
i'<i
}
}Y(i',j',i,j)\\
 & \leq(\frac{c\ln u}{v})^{2}\cdot v^{2}k^{2}\frac{\b\sigma^{-1}}{v^{1-o(1)}}\\
 & \leq\frac{\b\sigma^{-1}}{v^{1-o(1)}}.
\end{align*}

\subsection{Proof of Proposition \ref{prop:powers-to-sampler}}

Proposition \ref{prop:powers-to-sampler} directly follows from the
next two results.
\begin{prop}
For any $j',j\in[v]$, $i\in[u]$. $\E_{i'\in[u]}Y(i',j',i,j)\leq c\b\sigma^{-1}\frac{w}{v}$.
\end{prop}

\begin{proof}
We can write
\[
Y(i',j',i,j)=\sigma^{-F(G(i))}\mathbf{1}_{F(G(i))\neq0}
\]
where $F=\sum_{k\in[w]}f_{k}$ and for each $k\in[w]$, $f_{k}:[v]\to\{0,1\}$
is defined as follows:
\[
\begin{cases}
f_{k}(y)=1 & \text{if }y+j=G(i')_{k}+j'\bmod v;\\
f_{k}(y)=0 & \text{otherwise.}
\end{cases}
\]
 Note $\E[F]=w/v$. We can conclude since $G$ is a $(\sigma^{-1},\b)$
sampler. 
\end{proof}
\begin{prop}
\label{prop:small-to-all-new-2} Fix $1\leq k\leq u$ and non-negative
$X:[u]\times[u]\to\mathbb{R}$ s.t. for any $i\in[u]$, $\E_{i'\in[u]}X(i',i)\leq\mu$.
Then
\begin{align*}
\E_{\substack{T\subseteq[u]:\\
|T|=k
}
}\sum_{\substack{i',i\in T:\\
i'<i
}
}X(i',i) & \leq ck^{2}\mu\cdot
\end{align*}
\end{prop}

\begin{proof}
We have
\begin{align*}
\E_{\substack{T\subset[u]:\\
|T|=k
}
}\sum_{\substack{i',i\in T:\\
i'<i
}
}X(i',i) & =\binom{u}{k}^{-1}\binom{u-2}{k-2}\sum_{\substack{i',i\in[u]:\\
i'<i
}
}X(i',i)\leq\binom{u}{k}^{-1}\binom{u-2}{k-2}u^{2}\mu\leq ck^{2}\mu.
\end{align*}
\end{proof}
\begin{remark}
The condition $i'<i$ is necessary. Suppose $X(i',i)=u\mu$ for $i'=i$
and 0 otherwise. Then
\[
\E_{\substack{T\subseteq[u]:\\
|T|=k
}
}\sum_{\substack{i',i\in T}
}X(i',i)=ku\mu\cdot
\]
\end{remark}

\section{Balanced circuits matching KKL \label{sec:al-trick}}

In this section we prove Theorem \ref{thm:match-kkl}. We need the
following result due to Ajtai and Linial \cite{AL93}, the proof of
which we give at the end.
\begin{lem}
\label{prop:al-trick} Let $C_{1},C_{2}:\zo^{n}\to\zo$ be depth $d$
circuits such that $|\E[C_{1}]-1/2|<\e_{1}$, $\E[C_{2}]<(1-2\e_{1})$.
Then there is a depth $d+2$ circuit $C':\zo^{2n}\to\zo$ such that
$\E[C']=1/2$ and $I_{Q}(C')\leq I_{Q}(C_{1})+I_{Q}(C_{2})+\P[C_{2}=0]$
for any $Q\subseteq[2n].$ Moreover, if $C_{1},C_{2}$ are explicit
then $C'$ is explicit.
\end{lem}

\begin{proof}[Proof of Theorem \ref{thm:match-kkl}]
 Let $C_{1},C_{2}$ be read-once $OR_{v}$-$AND_{w},$ $OR_{v'}$-$AND_{w'}$
circuits where $vw,v'w'=cn$. First we set $w,v$ so that $v=\lceil2^{w}\ln2\rceil$.
Thus by the inequality $e^{-x/(1-x)}\leq1-x\leq e^{-x}$, 

\[
\left|\E[C_{1}]-1/2\right|\leq\frac{c_{1}\log n}{n}
\]
 for some fixed $c_{1}$. Similarly, we set $w',v'$ so that $v'=\lceil2^{w'}\ln\left(\frac{n}{3c_{1}\log n}\right)\rceil,w'$
which implies

\[
\left|\E[C_{2}]-\left(1-\frac{3c_{1}\log n}{n}\right)\right|\leq\frac{c\log^{3}}{n^{2}}.
\]
 Now for any $Q_{1},Q_{2}\subseteq[n]:|Q_{1}|=|Q_{2}|=1$ we have
\[
I_{Q_{1}}(C_{1})\leq\frac{c\log n}{n},
\]
\[
I_{Q_{2}}(C_{2})=(1-2^{-w'})^{v'-1}2^{-(w'-1)}\leq\frac{c\log n}{n}\cdot\frac{\log^{2}n}{n}=\frac{c\log^{3}n}{n^{2}}.
\]
We conclude by Lemma \ref{prop:al-trick}.
\end{proof}

\subsection{Proof of Lemma \ref{prop:al-trick}}

Let $\d=\E[C_{1}]-1/2$ and define $\mu$ so that $\E[C_{2}]=1-\mu$.
We define
\[
C'(x,y):=(C_{1}(x)\wedge C_{2}(y))\vee(D(x)\wedge\neg C_{2}(y))
\]
where $D(x):\zo^{n}\to\zo$ is an explicit DNF with $\E[D]=\frac{1}{2}+\d-\frac{\d}{\mu}\in[0,1]$.
The $\in$ follows since $\mu>2|\d|$ by hypothesis. We justify such
a $D$ in the end. Now,
\begin{align*}
\E[C'(x,y)] & =(1/2+\d)(1-\mu)+\E[D(x)]\cdot\mu=1/2.
\end{align*}
For any $Q\subseteq[2n]$, letting $Q_{1},Q_{2}$ denote the sets
in $x,y$ respectively, by a union bound we have $I_{Q}(C')\leq I_{Q_{1}}(C_{1})+I_{Q_{2}}(C_{2})+\P[C_{2}=0].$

To conclude, it remains to construct an explicit DNF $D$ such that
$\E[D]=2^{-n}\cdot k$ for any $k\in[2^{n}]$. Let us write the binary
representation of $k$ as $\alpha_{1}2^{n-1}+\dots+\alpha_{n}2^{0}$
where $\a_{i}\in\zo$. We construct $f$ by adding the term $\neg x_{1}\dots\neg x_{i-1}x_{i}$
if $\a_{i}=1$. The term $\neg x_{1}\dots\neg x_{i-1}x_{i}$ has $2^{n-i}$
inputs in its support. Furthermore, each term will be disjoint. Hence
$f$ attains the desired bias. $\hfill\qed$

\paragraph{Acknowledgments}

We are grateful to Raghu Meka for helpful discussions, providing clarifications
on \cite{Meka17}, and suggesting the use of tail bounds for expander
walks.

\bibliographystyle{alpha}
\bibliography{OmniBib}

\section{Deferred proofs}

\subsection{Proof of Fact \ref{fact:bias_independent}}

We use the inequalities $e^{-x(1+cx)}\leq e^{-\frac{x}{1-x}}\leq1-x\leq e^{-x}$
for $x\in(0,1)$, and $e^{x}\leq1+2x$ for $x\leq1$. First we lower
bound $bias(u,v,w,\sigma)$. We have
\[
(1-\sigma^{w})^{v}\leq e^{-v\sigma^{w}}\leq\ln2/u.
\]
Thus 
\[
bias(u,v,w,\sigma)\geq(1-\ln2/u)^{u}\geq e^{-\ln2(1+c/u)}\geq1/2-c/u.
\]
To upper bound $bias(u,v,w,\sigma)$, note

\[
(1-\sigma^{w})^{v}\geq e^{-\sigma^{w}v(1+c\sigma^{w})}\geq e^{-(\ln(u/\ln2)+\sigma^{w})(1+c\sigma^{w})}\geq(\ln2/u)(1-c\sigma^{w}).
\]
Thus 
\[
bias(u,v,w,\sigma)\le(1-(\ln2/u)(1-c\sigma^{w}))^{u}\leq e^{-\ln2(1-c\sigma^{w})}\leq1/2+c\sigma^{w}.
\]
From here we can conclude since $u^{-1}\leq v^{-1}\leq\sigma^{w}.$

\subsection{Proof of Claim \ref{claim:equivalent-definition}}

First suppose $\E_{i\in[u]}[\a^{F(G(i))}]\leq e^{c\b\a\mu}\leq1+c\b\a\mu.$
The last $\leq$ follows since $e^{x}\leq1+2x$ for $x\in[0,1]$.
Next, note that
\[
\E_{i\in[u]}[\a^{F(G(i))}\mathbf{1}_{F(G(i))\neq0}]=\E_{i\in[u]}[\a^{F(G(i))}]-\P_{i\in[u]}[F(G(i))=0].
\]
By Markov's inequality, $\P[F(G(i))=0]=1-\P[F(G(i))\geq1]\geq1-\E[F(G(i))]\geq1-c\b\a\mu$.
The last inequality follows by Jensen's inequality.

Now suppose $\E_{i\in[u]}[\a^{F(G(i))}\mathbf{1}_{F(G(i))\neq0}]\leq c\b\alpha\mu$.
Then since $1+x\leq e^{x}$, $\E_{i\in[u]}[\a^{F(G(i))}]\leq1+c\b\a\mu\leq e^{c\b\a\mu}$. 

\subsection{Proof of Fact \ref{fact:nonzero-to-zero}}

For the first inequality, let $\mu_{k}:=\E[f_{k}]$ for each $k\in[w]$.
Since $1+x\leq e^{x}$,
\[
\E[\a^{X}]=\prod_{k\in[w]}\E[\a^{X_{k}}]=\prod_{k\in[w]}(1+(\a-1)\mu_{k})\leq\prod_{k\in[w]}e^{(\a-1)\mu_{k}}=e^{(\a-1)\mu}.
\]
The second inequality follows by similar reasoning as in the proof
of Claim \ref{claim:equivalent-definition}.

\subsection{Proof of Proposition \ref{lem:janson}}

The statement follows by repeating the proof in \cite{AlonSpEr92}
(c.f. \cite{LSSUV21}) and avoiding the inequality $1+x\leq e^{x}$
in the end. We follow the presentation of \cite{LSSUV21} up until
equation (A.4), stated next. The probabilities below are over $D$,
and recall $L(C,i)=\sum_{\substack{i'\in[n]:i'<i,i'\sim i}
}\P[C_{i'}\wedge C_{i}],\Delta(C)=\sum_{i\in[n]}L(C,i)$. 
\begin{align*}
\P[C=0] & \leq\prod_{i\in[n]}\bigg(\P[C_{i}=0](1+2L(C,i))\bigg).
\end{align*}
Then by the AM-GM inequality,
\[
\prod_{i\in[n]}(1+2L(C,i))\leq\left(1+\frac{2\Delta(C)}{n}\right)^{n}\leq1+\sum_{\ell=1}^{n}\frac{2^{\ell}}{\ell!}\Delta(C)^{\ell}.
\]

\end{document}